\documentclass[prl,showpacs,preprint,amsmath,aps]{revtex4-1}
\usepackage{epsfig,graphicx}
\usepackage{amsmath}
\usepackage{mathrsfs}
\usepackage{url}
\usepackage{amsthm}
\usepackage{bm}
\usepackage{hyperref}

\hypersetup{colorlinks=true, citecolor=blue, urlcolor=blue, linkcolor=blue}
\newtheorem{theorem}{Theorem}

\begin{document}
\title{Tight upper bound on the quantum value of Svetlichny operators\\ under local filtering and hidden genuine nonlocality}
\author{Lingyun Sun$^1$}
\author{Li Xu$^1$}
\author{Jing Wang$^1$}
\author{Ming Li$^{1}$}
\author{Shuqian Shen$^1$}
\author{Lei Li$^1$}
\author{Shao-Ming Fei$^{2,3}$}
\affiliation{$^1$College of the Science, China University of
Petroleum, 266580 Qingdao, China\\
$^2$ School of Mathematical Sciences, Capital Normal University,
 100048 Beijing, China\\
$^3$ Max-Planck-Institute for Mathematics in the Sciences, 04103
Leipzig, Germany}

\date{\today}

\begin{abstract}
Nonlocal quantum correlations among the quantum subsystems play essential roles in quantum science.
The violation of the Svetlichny inequality provides sufficient conditions of genuine tripartite nonlocality. We provide tight upper bounds on the maximal quantum value of the Svetlichny operators under local filtering operations, and present a qualitative analytical analysis on the hidden genuine nonlocality for three-qubit systems. We investigate in detail two classes of three-qubit states whose hidden genuine nonlocalities can be revealed by local filtering.
\end{abstract}

\maketitle

\smallskip
\section{\leftline{Introduction}}
As important physical resources, quantum correlations like entanglement play fundamental roles in quantum information processing \cite{QE-2009-RevModPhys.81.865, New_QE-2013-Front.Phys, New_QE-2019-Front.Phys}, with numerous applications in quantum communication protocols with lower complexity \cite
{communication-complexity2002-PhysRevLett.89.197901,communication-complexity2010-RevModPhys.82.665} and higher security \cite{Security2001-PhysRevLett.87.117901,Security2006-PhysRevA.73.012314}.
Two systems $A$ and $B$ are entangled if the measurements on system $A$ does affect the probabilities
of the measurement outcomes from system $B$, and vise versa. For tripartite systems, there exit
correlations so called genuine entanglement \cite{Quantum_entanglement2011-PhysRevLett.106.250404}. A tripartite state may be genuine entangled even if
any pair of the subsystems are separable.

The stronger correlations than entanglement are nonlocal correlations.
Two systems $A$ and $B$ may be locally correlated even if they are entangled, as long as the correlations of their measurement outcomes can be described by classical correlation models of probability.
For a bipartite $\rho$, let $P(ab|XY)$ be the probability of measuring $X$ on subsystem $A$ with outcome $a$ and $Y$ on subsystem $B$ with outcome $b$. If the probability correlation $P(ab|XY)$ can be expressed in the form,
$P(ab|XY)=\sum\limits_{\lambda} {q_{\lambda}}P_{\lambda}(a|X)P_{\lambda}(b|Y)$,
where $\lambda$ is regarded as a shared local hidden variable, $q_{\lambda}\geq0$, and $\Sigma_{\lambda}q(\lambda)=1$, then the state $\rho$ is regarded as locally correlated, and admits
a local hidden-variable (LHV) model.
The bipartite Bell nonlocality \cite{Bell-1964-Phys.1.195,Bell-nonlocality2014-RevModPhys.86.419} can be witnessed by the violation of Bell inequalities \cite{Bell-nonlocality2014-RevModPhys.86.419}.

Similar to the quantum entanglement, the quantum nonlocality becomes subtler for multipartite and high dimensional systems \cite{Def-MNL-2013-PhysRevA.88.014102,New_nonlocilty-FOP2012, New_nonlocilty-FOP2018}.
Let $\rho$ be a tripartite state. Performing local measurements $X$, $Y$ and $Z$ on the subsystems $A$, $B$ and $C$ with outcomes $a$, $b$ and $c$, respectively, we say the state is three local if the corresponding probability correlations $P(abc|XYZ)$ can be written as
\begin{equation}\label{3LHV}
P(abc|XYZ)=\sum\limits_{\lambda} {q_{\lambda}}P_{\lambda}(a|X)P_{\lambda}(b|Y)P_{\lambda}(c|Z),
\end{equation}
where $0\leq q_{\lambda}\leq 1$ and $\Sigma_{\lambda} q_{\lambda}=1$.
Otherwise, the state is called nonthree or full local.
A nonthree local state is said to be hybrid-nonlocal, admitting bi-LHV model, if
\begin{equation}\label{bi-LHV}
P(abc|XYZ)=\sum\limits_{\lambda} {q_{\lambda}}P_{\lambda}(ab|XY)P_{\lambda}(c|Z)
    +\sum\limits_{\mu}{q_{\mu}}P_{\mu}(ac|XZ)P_{\mu}(b|Y)
    +\sum\limits_{\upsilon}{q_{\upsilon}}P_{\upsilon}(bc|YZ)P_{\upsilon}(a|X),
\end{equation}
where $0\leq q_{\lambda},q_{\mu},q_{\upsilon}\leq 1$, and $\Sigma_{\lambda} q_{\lambda}+\Sigma_{\mu} q_{\mu}+\Sigma_{\upsilon} q_{\upsilon}=1$.

If the probability correlation can not be written in form (\ref{bi-LHV}), the state is called genuine tripartite nonlocal. The genuine tripartite nonlocality of a state can be detected by the
Svetlichny inequality (SI) \cite{Svetlichny1987-PhysRevD.35.3066}. The violation of SI is a sufficient condition for the genuine tripartite nonlocality. However, generally it is not easy to verify
such violations. In \cite{LM2017-PhysRevA.96.042323}, the authors presented a tight upper bound for the maximal quantum value of the Svetlichny operator.

For bipartite case, it has been shown that the nonlocality of certain quantum states can be revealed by using local filters before performing a standard Bell test, known as genuine hidden nonlocality \cite{localfilter-2013-PhysRevLett.111.160402}.
In \cite{Max-entanglement-PhysRevA.68.012103}, Verstraete et. al demonstrated that the optimal local filtering operations can maximize certain entanglement measures.
Moreover, quantum properties such as Bell nonlocality and steerability \cite{Steering2007-PhysRevLett.98.140402} of specific quantum states can be revealed by local filtering.
In \cite{lf-CHSH-PhysRevLett.74.2619,PhysLettA.210.151}, the authors considered two-qubit states which do not violate CHSH inequality before, but do violate after performing local filtering operations.
The maximal violation of the CHSH inequality and the lower bound of the maximal violation of V$\acute{e}$rtesi inequality under local filtering operations were computed analytically in \cite{li2017maximal}.
In \cite{Hidden-steerability2019-PhysRevA.99.030101}, Pramanik et al. showed that there exist initially unsteerable bipartite states which show steerability after local filtering.

For a tripartite state $\rho$, under local filtering transformations one gets,
\begin{equation}
\label{Local filtered state}
{\rho }'=\frac{1}{N}\left ( F_{A}\otimes F_{B}\otimes F_{C}\right )\rho\left ( F_{A}\otimes F_{B}\otimes F_{C}\right )^{\dagger},
\end{equation}
where $N={\rm tr}\left[\left (F _{A} \otimes F_{B}\otimes F_{C} \right )\rho \left (F _{A} \otimes F_{B}\otimes F_{C} \right )^{\dagger} \right ]$ is a normalization factor, $F_{A}$, $F_{B}$ and $F_{C}$ are positive operators acting on the local subsystems, respectively.
In \cite{AON2020-PhysRevLett.124.050401}, Tendick et al. discussed the relation between entanglement and nonlocality in the hidden nonlocality scenario, and presented a fully biseparable three-qubit
bound entangled state with a local model for the most general measurements. By using the ${\rm \acute{S}}$liwa's inequality and an iterative sequence of semidefinite programs it is shown that the local model breaks down when suitable local filters are applied, which demonstrates the activation of nonlocality in bound entanglement, as well as that genuine hidden nonlocality does not imply entanglement distillability.

In this paper, we first study the maximal quantum value of the Svetlichny operators after local filtering operations for any three-qubit system.
A tight upper bound for the maximal value of the Svetlichny operators after local filtering is obtained. Then we take the color noised Greenberger-Horne-Zeilinger (GHZ)-class states as examples to illustrate how local filter operations work in nonlocality improvement.
We show that the hidden genuine nonlocalities can be revealed by local filtering
for these classes of three-qubit states.

\section{Tight upper bound on the value of Svetlichny operator under local filtering}
The Svetlichny operator in Svetlichny inequality \cite{Svetlichny1987-PhysRevD.35.3066} reads
\begin{equation}
\label{Svetlichny operator}
\mathcal{S}=A\otimes [(B+B')\otimes C+(B-B')\otimes C']+A'\otimes [(B-B')\otimes C-(B+B')\otimes C'],
\end{equation}
where $A,A',B,B',C$ and $C'$ denote the local observables of the form $G=\vec{g}\cdot \vec{\sigma}=\Sigma_{k=1}^3 g_{k}\sigma_{k}$, $G\!\in\!\{A,A',B,B',C,C'\}$ and $\vec{g}\!\in\!\{\vec{a},\vec{a}\,',\vec{b},\vec{b}\,',\vec{c},\vec{c}\,'\}$, respectively. $\vec{\sigma}=(\sigma_{1},\sigma_{2},\sigma_{3})$ with $\sigma_{i}$, $i=1,2,3$, the standard Pauli matrices. $\vec{g}$ is a three-dimensional real unit vector.

The mean value of the Svetlichny operator for an arbitrary three-qubit state $\rho$ admitting a bi-LHV model satisfies the following inequality \cite{Svetlichny1987-PhysRevD.35.3066},
\begin{equation}\label{SI}
|\langle\mathcal{S}\rangle_{\rho}|\leq 4,
\end{equation}
where $\langle\mathcal{S}\rangle_{\rho}$={\rm tr}$(\mathcal{S}\rho)$.
A state violating the inequality (\ref{SI}) is called genuine three-qubit nonlocal.
It has been shown that the maximal quantum value of the Svetlichny operator for three-qubit systems is upper bounded \cite{LM2017-PhysRevA.96.042323},
\begin{equation}
\label{new upper bound}
\mathcal{Q}(\mathcal{S})\equiv {\rm max}|\langle\mathcal{S}\rangle_{\rho}|\leq4\lambda_{1},
\end{equation}
where $\lambda_{1}$ is the maximal singular value of the matrix $M=(M_{j,ik})$, with $M_{ijk}$={\rm tr}$[\rho(\sigma_{i}\otimes\sigma_{j}\otimes\sigma_{k})],i,j,k=1,2,3.$
The upper bound is tight if the degeneracy of $\lambda_{1}$ is more than 1, and the two degenerate nine-dimensional singular vectors corresponding to $\lambda_{1}$ take the form of $\vec{a}\otimes\vec{c}-\vec{a}\,'\otimes\vec{c}\,'$ and $\vec{a}\otimes\vec{c}\,'+\vec{a}\,'\otimes\vec{c}$.

Let $F_{A}=U\Sigma _{A}U^{\dagger}$, $F_{B}=V\Sigma _{B}V^{\dagger}$ and $F_{C}=W\Sigma _{C}W^{\dagger}$ be the spectral decompositions of the filter operators $F_{A}$, $F_{B}$ and $F_{C}$, respectively, where $U$, $V$ and $W$ are unitary operators. Set
$\delta_{l}=\Sigma _{A}\sigma_{l}\Sigma _{A}$, $\eta_{m} =\Sigma _{B}\sigma_{m}\Sigma _{B}$ and $\gamma_{n}=\Sigma _{C}\sigma_{n}\Sigma _{C}$.
Without loss of generality, we assume that
\begin{center}
$\Sigma _{A}=\begin{pmatrix}
x  & 0\\
0 & 1
\end{pmatrix}$,
$\Sigma _{B}=\begin{pmatrix}
y  & 0\\
0 & 1
\end{pmatrix}$ and
$\Sigma _{C}=\begin{pmatrix}
z  & 0\\
0 & 1
\end{pmatrix}$
\end{center}
with $x,y,z\geq 0$.
Let $X=(X_{m,ln})$ be a matrix with entries given by
\begin{equation}\label{Xmatrix}
X_{lmn}={\rm tr}[\varrho(\delta_{l}\otimes \eta_{m}\otimes\gamma_{n})],~~~l,m,n=1,2,3,
\end{equation}
where $\varrho$ is any state that is locally unitary equivalent to $\rho$.

\begin{theorem}\label{theorem 1}
\rm For the local filtered quantum state ${\rho }'=\frac{1}{N}\left ( F_{A}\otimes F_{B}\otimes F_{C}\right )\rho\left ( F_{A}\otimes F_{B}\otimes F_{C}\right )^{\dagger}$ of a three-qubit $\rho$, the maximal quantum value of the Svetlichny operator $\mathcal{S}$ defined in Eq. (\ref{Svetlichny operator}) satisfies
\begin{equation}
\label{th1-formula}
\mathcal{Q}(\mathcal{S})'={\rm max}\left |\left \langle \mathcal{S} \right \rangle _{{\rho}'}\right |\leq 4{\lambda_{1}}',
\end{equation}
where $\left \langle \mathcal{S} \right \rangle _{{\rho}'}={\rm tr}(\mathcal{S}{\rho}')$, ${\lambda_{1}}'$ is the maximal singular value of the matrix $X/N$, with $X$ defined in Eq. (\ref{Xmatrix}), taking over all quantum states $\varrho$ which are locally unitary equivalent to $\rho$. Equivalently, ${\lambda_{1}}'$ is also the maximal singular value of the matrix ${M}'=({M_{j,ik}}')$, with ${M_{ijk}}'={\rm tr}[{\rho}'(\sigma _{i}\otimes\sigma _{j}\otimes\sigma _{k})]$, $i,j,k=1,2,3.$
\end{theorem}

\begin{proof}
The normalization factor $N$ has the following form,
\begin{align*}
 N&={\rm tr}\left [ (U\Sigma_{A}^{2}U^{\dagger}\otimes V\Sigma_{B}^{2}V^{\dagger}\otimes W\Sigma_{C}^{2}W^{\dagger})\rho   \right ] \\
 &={\rm tr}\left [(\Sigma_{A}^{2}\otimes \Sigma_{B}^{2}\otimes \Sigma_{C}^{2})(U^{\dagger}\otimes V^{\dagger}\otimes W^{\dagger})\rho (U\otimes V\otimes W) \right ]\\
 &={\rm tr}\left [ (\Sigma_{A}^{2}\otimes \Sigma_{B}^{2}\otimes \Sigma_{C}^{2})\varrho \right],
\end{align*}
where $\varrho =(U^{\dagger}\otimes V^{\dagger}\otimes W^{\dagger})\rho (U\otimes V\otimes W)$. Since $\rho$ and $\varrho$ are local unitary equivalent, they have the same value of the maximal violation of the SI.

From the double cover relationship \cite{SO(3)-1995-PhysRevA.52.4396,SU(3)-LM-2014-PhysRevA.89.062325} between the special unitary group $SU(2)$ and the special orthogonal group $SO(3)$,
$U\sigma_{_{i}} U^{\dagger}=\sum_{j=1}^{3}O_{ij}\sigma_{j}$,
where $U$ is any given unitary operator and the matrix $O$ with entries $O_{ij}$ belongs to $SO(3)$,
we have
\begin{align}\label{proof of th1}
\notag {M_{ijk}}'&={\rm tr}[{\rho}'(\sigma_{i}\otimes \sigma_{j}\otimes \sigma_{k})]\\
\notag &=\frac{1}{N}{\rm tr}\left [ \left (F_{A}\otimes F_{B}\otimes F_{C} \right )\rho\left ( F_{A}^{\dagger}\otimes F_{B}^{\dagger}\otimes F_{C}^{\dagger}\right )(\sigma_{i}\otimes \sigma_{j}\otimes \sigma_{k}) \right ]\\
\notag &=\frac{1}{N}{\rm tr}\left [ \rho (U \Sigma _{A}U^{\dagger}\sigma_{i}U \Sigma _{A}U^{\dagger}\otimes V \Sigma _{B}V^{\dagger}\sigma_{j}V \Sigma _{B}V^{\dagger}\otimes W \Sigma _{C}W^{\dagger}\sigma_{k}W \Sigma _{C}W^{\dagger})  \right ]\\
\notag &=\frac{1}{N}\sum_{l,m,n}{\rm tr}\left [ (U^{\dagger}\otimes V^{\dagger}\otimes W^{\dagger})\rho (U\otimes V\otimes W)(\Sigma_{A}O_{il}^{A}\sigma_{l} \Sigma_{A}\otimes \Sigma_{B}O_{jm}^{B}\sigma_{m}\Sigma_{B}\otimes \Sigma_{C}O_{kn}^{C}\sigma_{n}\Sigma_{C})\right ]\\
\notag &=\frac{1}{N}\sum_{l,m,n}O_{il}^{A}O_{jm}^{B}O_{kn}^{C}{\rm tr}\left [ \varrho (\Sigma_{A}\sigma_{l} \Sigma_{A}\otimes \Sigma_{B}\sigma_{m}\Sigma_{B}\otimes \Sigma_{C}\sigma_{n}\Sigma_{C})  \right ]\\
\notag &=\frac{1}{N}\sum_{l,m,n}O_{il}^{A}O_{jm}^{B}O_{kn}^{C}{\rm tr}\left [ \varrho (\delta_{l}\otimes\eta_{m}\otimes\gamma_{n}) \right ]\\
&=\frac{1}{N}\left [ O_{A}X\left ( O_{B}^{T}\otimes O_{C}^{T} \right ) \right ]_{ijk}.
\end{align}

Therefore, we have ${M}'=\left [ O_{A}X\left ( O_{B}^{T}\otimes O_{C}^{T}\right ) \right]/N$, and
\begin{equation}
\left ( {M}' \right )^{\dagger}{M}'=\frac{1}{N^{2}}\left ( O_{B}\otimes O_{C} \right )X^{\dagger}O_{A}^{\dagger}O_{A}X\left ( O_{B}\otimes O_{C} \right )^{\dagger}
=\frac{1}{N^{2}}\left (O_{B}\otimes O_{C}  \right )X^{\dagger}X\left ( O_{B}\otimes O_{C} \right )^{\dagger}.
\end{equation}
By noticing the orthogonality of the operator $O_{B}\otimes O_{C}$, one obtains that ${M}'\left ( {M}' \right )^{\dagger}$ has the same eigenvalues as $X^{\dagger}X/N^{2}$. Hence, $M'$ has the same singular values as $X/N$. Let $\vec{v}$ be a nine-dimensional singular vector of the matrix $X/N$. Then $(O_{B}\otimes O_{C})\vec{v}$ is a nine-dimensional singular vector of the matrix $M'$.
\end{proof}

\section{\leftline{Tightness of the upper bound and hidden genuine nonlocality}}
As applications of the Theorem 1, we consider the activation of the hidden genuine nonlocality
of three-qubit systems. We present two classes of three-qubit states which admit bi-LHV model before
local filtering, but display genuine nonlocality after local filtering.

Let us begin with the two-qubit isotropic states
\begin{equation}
{\chi}_{iso}(p)=p|\phi\rangle\langle\phi|+(1-p)\frac{I_{4}}{4},
\end{equation}
where $|\phi\rangle=\frac{1}{\sqrt{2}}(|00\rangle+|11\rangle)$ is the maximally entangled state, $I_{4}$ denotes the $4\times 4$ identity matrix, and $0\leq p\leq 1$.
The state ${\chi}_{iso}(p)$ is fully local for $0\leq p\leq 0.4167$ and the related local model can be generalized from isotropic states to the following states \cite{LHV-PhysRevLett.99.040403},
\begin{equation}
\hat{\chi}(p,\theta)=p|\psi_{s}\rangle\langle\psi_{s}|+(1-p)\frac{I_{4}}{4},
\end{equation}
where $|\psi_{s}\rangle={\rm cos}\,\theta|00\rangle+{\rm sin}\,\theta|11\rangle$, and $0\leq\theta\leq\pi/4$.

Consider the following states,
\begin{equation}
\chi(p,\theta)=p|\psi_{s}\rangle\langle\psi_{s}|+(1-p)|0\rangle\langle0|\otimes\frac{I_{2}}{2},
\end{equation}
where $I_{2}$ denotes the $2\times 2$ identity matrix, $0\leq\theta\leq\pi/4$, and $0\leq p\leq 1$.
Following the protocol presented in \cite{localfilter-2013-PhysRevLett.111.160402}, we have that
$\chi(p,\pi/4)$ admits an LHV model as $\hat{\chi}(p,\pi/4)$.
Namely, $\chi(p,\theta)$, with $0\leq\theta\leq\pi/4$, admits an LHV model for $0\leq p\leq 0.4167$.

Any bipartite local states can be converted to multipartite states with a bi-local model \cite{construction-PhysRevLett.115.030404}.
Following the construction given in \cite{construction-PhysRevLett.115.030404}, we can
transform the states $\chi(p,\theta)$ into the mixture of the colored noise and the three-qubit GHZ-class states,
\begin{equation}\label{state1}
  \rho_{\chi}(p,\theta)=p|\Psi_{s}\rangle\langle \Psi_{s}|+(1-p)|00\rangle\langle 00|\otimes\frac{I_{2}}{2},
\end{equation}
where $|\Psi_{s}\rangle={\rm cos}\,\theta|000\rangle+{\rm sin}\,\theta|111\rangle$. Analogously, the state $\rho_{\chi}(p,\theta)$ admits an bi-LHV model for $0\leq p\leq 0.4167$.
In the following, we set $\theta=\pi/8$ and consider the activation of the hidden genuine nonlocality
of $\rho_{\chi}(p,\pi/8)$ under local filtering.

Firstly, based on the genuine multipartite concurrence of three-qubit \emph{X} states \cite{X.GME-2012-PhysRevA.86.062303}, one can show that $\rho_{\chi}(p,\pi/8)$ is genuine multipartite entangled for $0<p\leq 1$.
Secondly, the quantum state $\rho_{\chi}(p,\pi/8)$ attains the upper bound on the mean values of the SI operators, but never violates the SI, which can be seen from the matrix $M$ of $\rho_{\chi}(p,\pi/8)$ defined in (\ref{new upper bound}),
\begin{equation}
M=\left(
    \begin{array}{ccccccccc}
   \displaystyle \frac{\sqrt{2}p}{2} & 0 & 0 & 0 & \displaystyle -\frac{\sqrt{2}p}{2} & 0 & 0 & 0 & 0\\
  0 & \displaystyle -\frac{\sqrt{2}p}{2} & 0 & \displaystyle -\frac{\sqrt{2}p}{2} & 0 & 0 & 0 & 0 & 0\\
      0 & 0 & 0 & 0 & 0 & 0 & 0 & 0 & \displaystyle \frac{\sqrt{2}p}{2}
    \end{array}
  \right).
\end{equation}
The singular values of the matrix $M$ are $p$, $p$ and $\displaystyle \frac{\sqrt{2}p}{2}$. Hence, $\lambda_{1}=p$. The upper bound of the maximal mean value of the Svetlichny operator is
$\mathcal{Q}(\mathcal{S})={\rm max}|\langle\mathcal{S}\rangle_{\rho_{\chi}(p,\pi/8)}|\leq 4\lambda_{1}=4p$.
In order to assure that the upper bound can be used to determine the violation of the SI, one needs to prove that the bound is attained for the state $\rho_{\chi}(p,\pi/8)$, which requires that two nine-dimensional singular vectors have the forms $\vec{a}\otimes\vec{c}-\vec{a}\,'\otimes\vec{c}\,'$ and
$\vec{a}\otimes\vec{c}\,'+\vec{a}\,'\otimes\vec{c}$ exist.
We select the two singular vectors corresponding to the degenerated $\lambda_{1}$ as $\vec{v_{1}}=(1,0,0,0,-1,0,0,0,0)^{T}=(1,0,0)^{T}\otimes (1,0,0)^{T}-(0,-1,0)^{T}\otimes (0,-1,0)^{T}$ and $\vec{v_{2}}=(0,-1,0,-1,0,0,0,0,0)^{T}=(1,0,0)^{T}\otimes (0,-1,0)^{T}+(0,-1,0)^{T}\otimes (1,0,0)^{T}$.
Setting $\vec{a}=(1,0,0)^{T}$, $\vec{a}\,'=(0,-1,0)^{T}$, $\vec{c}=(1,0,0)^{T}$ and $\vec{c}\,'=(0,-1,0)^{T}$, and choosing $\vec{b}$ and $\vec{b}\,'$ to be suitable unit vectors of proper measurement directions of $B$ and $B'$ in Eq. (\ref{Svetlichny operator}), we can show that
the upper bound $4p$ is attained for $\rho_{\chi}(p,\pi/8)$. Nevertheless, the violation of the SI never happens for the quantum state $\rho_{\chi}(p,\pi/8)$ as expected.

Now we consider the local filtering of $\rho_{\chi}(p,\pi/8)$. By direct computation, we have the matrix $\tilde{M}=(\tilde{M}_{m,ln})=({\rm tr}[\rho_{\chi}(p,\pi/8)(\delta_{l}\otimes\eta_{m}\otimes\gamma_{n})])$, $l,m,n=1,2,3$,
\begin{equation}
\tilde{M}=\left(
    \begin{array}{ccccccccc}
      \displaystyle\frac{\sqrt{2}pxyz}{2} & 0 & 0 & 0 & \displaystyle-\frac{\sqrt{2}pxyz}{2} & 0 & 0 & 0 & 0\\
      0 & \displaystyle -\frac{\sqrt{2}pxyz}{2} & 0 & \displaystyle -\frac{\sqrt{2}pxyz}{2} & 0 & 0 & 0 & 0 & 0\\
      0 & 0 & 0 & 0 & 0 & 0 & 0 & 0 & D\\
    \end{array}
  \right),
\end{equation}
where $D=\displaystyle -\frac{2-\sqrt{2}}{4}p-\frac{1-p}{2}x^{2}y^{2}+\frac{2+\sqrt{2}p}{4}x^{2}y^{2}z^{2}$.
The singular values of the matrix $\tilde{M}$ are $pxyz$, $pxyz$ and $D.$
Since $\rho_{\chi}(p,\pi/8)$ and $\varrho_{\chi}(p,\pi/8)$ are locally unitary equivalent, we conclude that ${pxyz}/{N}$, ${pxyz}/{N}$ and ${D}/{N}$ are the singular values of the matrix $X/N$, where $N={\rm tr} [\rho_{\chi}(p,\pi/8)(\Sigma^{2}_{A}\otimes\Sigma^{2}_{B}\otimes\Sigma^{2}_{C})]
=\displaystyle \frac{2-\sqrt{2}}{4}p+\frac{1-p}{2}x^{2}y^{2}+\frac{2+\sqrt{2}p}{4}x^{2}y^{2}z^{2}$, which
are also the singular values of the matrix $M'$ in Theorem 1.
The maximal singular value $\lambda_{1}'$ is ${pxyz}/{N}$ for given $p$, with ${pxyz}/{N}>{D}/{N}$.
Then the upper bound of the maximal value of the Svetlichny operator is given by
\begin{equation}
\mathcal{Q}(\mathcal{S})'={\rm max}|\langle\mathcal{S}\rangle_{\rho_{\chi}'(p,\pi/8)}|\leq 4\lambda_{1}'=\displaystyle\frac{4pxyz}{N}.
\end{equation}

The matrix $X/N$ also has the singular vectors $\vec{v_{1}}$ and $\vec{v_{2}}$ with respect to the singular value $\lambda_{1}'$. According to Theorem 1, the singular vectors of $M'$ corresponding to $\lambda_{1}'$ are $(O_{B}\otimes O_{C})\vec{v_{1}}=O_{B}\,\vec{a}\otimes O_{C}\,\vec{c}-O_{B}\,\vec{a}\,'\otimes O_{C}\,\vec{c}$ and $(O_{B}\otimes O_{C})\vec{v_{2}}=O_{B}\,\vec{a}\otimes O_{C}\,\vec{c}\,'+O_{B}\,\vec{a}\,'\otimes O_{C}\,\vec{c}$,
where $O_{B}$ and $O_{C}$ belong to $SU(3)$.
The upper bound is saturated as the singular vectors can be written in required decomposition forms.
The SI is violated if and only if $\lambda_{1}'=\displaystyle \frac{pxyz}{N}>1$.
Maximizing $\lambda_{1}'$ under the restriction $\displaystyle\frac{pxyz}{N}>\frac{D}{N}$,
we obtain that the quantum states $\rho_{\chi}'(p,\pi/8)$ violate the SI and are genuine three-qubit
nonlocal for $0.3697\leq p\leq 1$, although $\rho_{\chi}(p,\pi/8)$ is bi-local, see Fig. 1.

\begin{figure}[htbp]
 \centering
  \includegraphics[width=15cm,height=5cm]{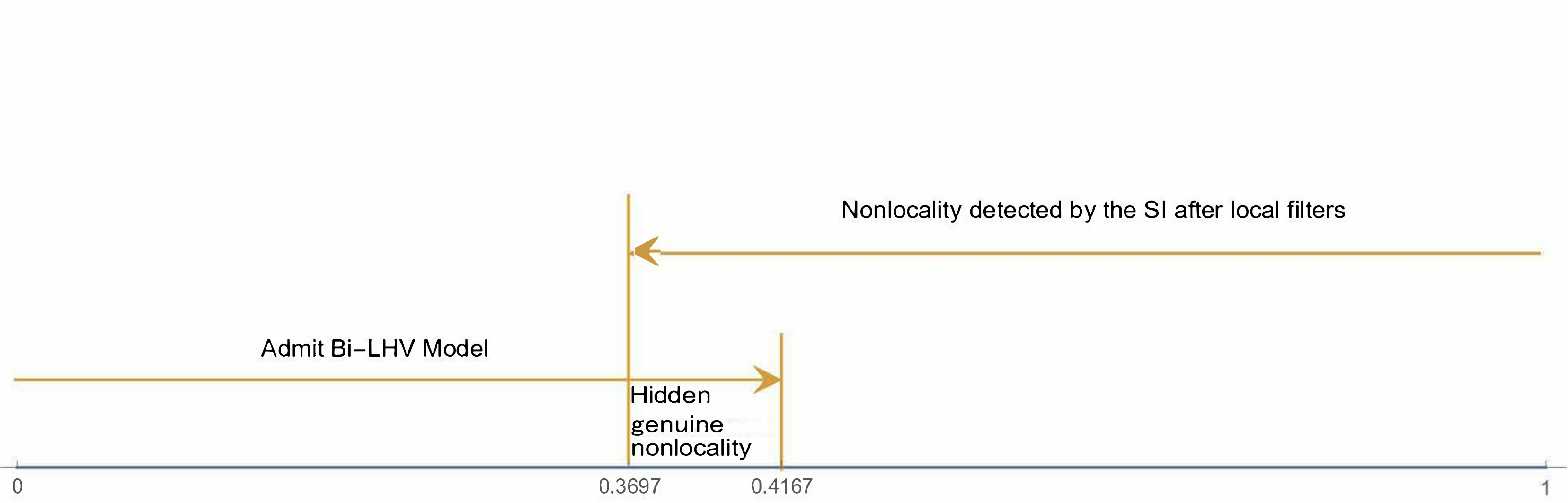}
  \caption{The state $\rho_{\chi}(p,\pi/8)$ admits a bi-local hidden model for $0\leq p\leq 0.4167$ and never violates SI for $0\leq p\leq 1$ as the upper bound in Theorem 1 is saturated. The locally filtered state shows the genuine nonlocality for $0.3697\leq p\leq 1$. The hidden genuine tripartite nonlocality is revealed for $0.3697\leq p\leq 0.4167$.}
\end{figure}

Now we consider another class of three-qubit states.
Consider the mixture of the three-qubit GHZ states and the colored noise,
\begin{equation}\label{state2}
\rho=p|{\rm GHZ}\rangle\langle {\rm GHZ}|+\frac{1-p}{4}\widetilde{I_{0}}\otimes I_{4},
\end{equation}
where $|{\rm GHZ}\rangle=\frac{1}{\sqrt{2}}(|000\rangle+|111\rangle)$, $\widetilde{I_{0}}=\left(
                           \begin{array}{cc}
                             1 & 0 \\
                             0 & 0 \\
                           \end{array}
                         \right)$ and $0\leq p\leq 1$.
The state $\rho$ is genuine multipartite entangled for $0<p\leq 1$ by using the criterion given in \cite{X.GME-2012-PhysRevA.86.062303}. The corresponding matrix $M$,
\begin{equation}
M=\left(
    \begin{array}{ccccccccc}
      p & 0 & 0 & 0 & -p & 0 & 0 & 0 & 0\\
      0 & -p & 0 & -p & 0 & 0 & 0 & 0 & 0\\
      0 & 0 & 0 & 0 & 0 & 0 & 0 & 0 & 0\\
    \end{array}
  \right),
\end{equation}
has singular values $\sqrt{2}p$, $\sqrt{2}p$ and 0, i.e., $\lambda_{1}=\sqrt{2}p$.
The upper bound of the maximal value of the Svetlichny operator satisfies
$\mathcal{Q}(\mathcal{S})=\rm{max}|\langle\mathcal{S}\rangle_{\rho}|\leq 4\lambda_{1}=4\sqrt{2}p$.
This upper bound is saturated since two nine-dimensional singular vectors of the forms, $\vec{a}\otimes\vec{c}-\vec{a}\,'\otimes\vec{c}\,'$ and $\vec{a}\otimes\vec{c}\,'+\vec{a}\,'\otimes\vec{c}$ can be found in the following way. Take the singular vectors corresponding to $\lambda_{1}$ to be $\vec{v_{1}}=(1,0,0,0,-1,0,0,0,0)^{T}$ and $\vec{v_{2}}=(0,-1,0,-1,0,0,0,0,0)^{T}$, which have exactly the forms, $(1,0,0)^{T}\otimes (1,0,0)^{T}-(0,-1,0)^{T}\otimes (0,-1,0)^{T}$ and
$(1,0,0)^{T}\otimes (0,-1,0)^{T}+(0,-1,0)^{T}\otimes (1,0,0)^{T}$, respectively. By defining
$\vec{a}=(1,0,0)^{T}$, $\vec{a}\,'=(0,-1,0)^{T}$, $\vec{c}=(1,0,0)^{T}$ and $\vec{c}\,'=(0,-1,0)^{T}$,
and selecting suitable $\vec{b}$ and $\vec{b}\,'$, the upper bound is attained. Therefore, the state $\rho$ in Eq. (\ref{state2}) violates the SI if and only if $0.707107<p\leq 1$.

The matrix $\tilde{M}=(\tilde{M}_{m,ln})=({\rm tr}[\rho(\delta_{l}\otimes\eta_{m}\otimes\gamma_{n})])$, $l,m,n=1,2,3$, has the form,
\begin{equation}
\tilde{M}=\left(
    \begin{array}{ccccccccc}
      pxyz & 0 & 0 & 0 & -pxyz & 0 & 0 & 0 & 0\\
      0 & -pxyz & 0 & -pxyz & 0 & 0 & 0 & 0 & 0\\
      0 & 0 & 0 & 0 & 0 & 0 & 0 & 0 & D
    \end{array}
  \right),
\end{equation}
where $D=-\frac{1}{2}p+\frac{1-p}{4}x^{2}-\frac{1-p}{4}x^{2}y^{2}-\frac{1-p}{4}x^{2}z^{2}
+\frac{1+p}{4}x^{2}y^{2}z^{2}.$
The singular values of $\tilde{M}$ are $\sqrt{2}pxyz,$ $\sqrt{2}pxyz$ and $D.$
Due to the local unitary equivalence between $\rho$ and $\varrho$,
the singular values of the matrix $X/N$ are
${\sqrt{2}pxyz}/{N}$, ${\sqrt{2}pxyz}/{N}$ and ${D}/{N}$, where $N={\rm tr}[\rho(\Sigma^{2}_{A}\otimes\Sigma^{2}_{B}\otimes\Sigma^{2}_{C})]
=\frac{1}{2}\,p+ \frac{1-p}{4}x^{2}+\frac{1-p}{4}x^{2}y^{2}+\frac{1-p}{4}x^{2}z^{2}+\frac{1+p}{4}x^{2}y^{2}z^{2}.$
According to Theorem 1, these values are also the singular values of the matrix $M'$.
We have $\lambda_{1}'=\frac{\sqrt{2}pxyz}{N}$ for $\frac{\sqrt{2}pxyz}{N}>\frac{D}{N}$.
The matrix $X/N$ also has the same singular vectors $\vec{v_{1}}$ and $\vec{v_{2}}$ with respect to $\lambda_{1}'$. The singular vectors of $M'$ with respect to $\lambda_{1}'$ are $(O_{B}\otimes O_{C})\vec{v_{1}}$ and $(O_{B}\otimes O_{C})\vec{v_{2}}$, namely,
$O_{B}\,\vec{a}\otimes O_{C}\,\vec{c}-O_{B}\,\vec{a}\,'\otimes O_{C}\,\vec{c}\,'$ and  $O_{B}\,\vec{a}\otimes O_{C}\,\vec{c}\,'+O_{B}\,\vec{a}\,'\otimes O_{C}\,\vec{c},$
with $O_{B}$ and $O_{C}$ belonging to $SU(3)$.
Hence the upper bound is also saturated for the locally filtered state.
Therefore, the state violates the SI if and only if $\lambda_{1}'=\frac{\sqrt{2}pxyz}{N}>1$.
Then the upper bound of the maximal value of the Svetlichny operator satisfies
\begin{equation}
  \mathcal{Q}(\mathcal{S})'={\rm max}|\langle\mathcal{S}\rangle_{\rho}|\leq 4\lambda_{1}'=\displaystyle \frac{4\sqrt{2}pxyz}{N}.
\end{equation}
Based on the above analysis, the genuine nonlocality of the quantum state $\rho'$ is detected by the SI for $0.3334\leq p\leq 1$. Therefor, the hidden genuine nonlocality of $\rho$ is revealed by local filtering operations for $0.3334\leq p\leq 0.7071$, see Fig. 2.

\begin{figure}[htbp]
\label{figure2}
  \centering
  \includegraphics[width=3in]{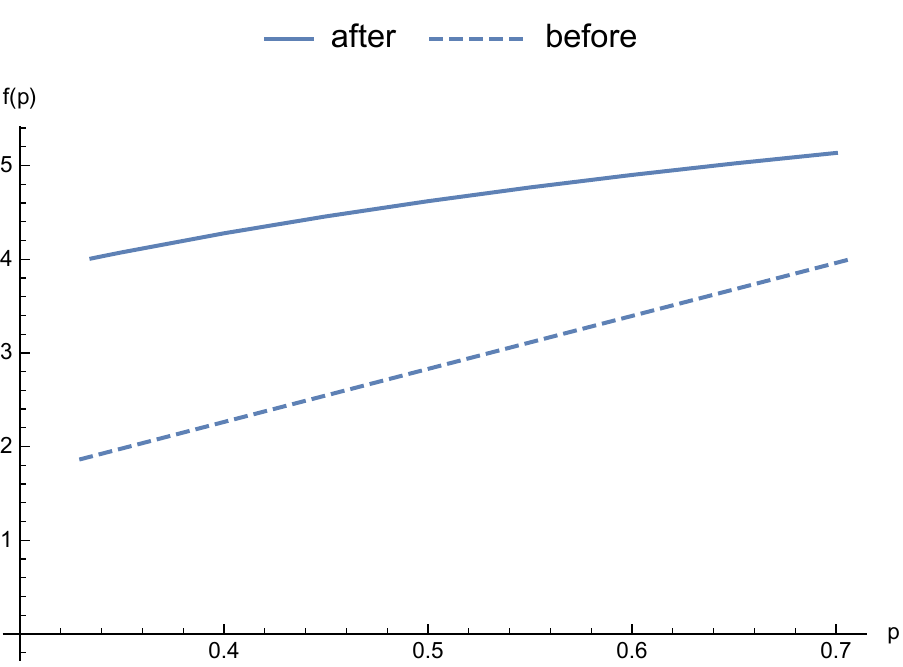}
  \caption{
  Denote $f(p)$ the maximal value of $\mathcal{Q}(\mathcal{S})$:
  Max$|\langle\mathcal{S}\rangle_{\rho}|$ (dashed line) max$|\langle\mathcal{S}\rangle_{\rho'}|$ (solid line). As the upper bound in Theorem 1 is saturated, for $0.3334\leq p\leq 0.7071$ the state $\rho=p\,|{\rm GHZ}\rangle\langle {\rm GHZ}|+\frac{1-p}{4}\widetilde{I_{0}}\otimes I_{4}$ does not violate the SI, but its locally filtered state $\rho'$ shows genuine nonlocality.}
\end{figure}

\section{\leftline{Conclusions and Discussions}}
\label{sec::Conclusions and Discussions}
We have presented a qualitative analytical analysis of the hidden genuine nonlocality for three-qubit systems by providing a tight upper bound on the maximal quantum value of the Svetlichny operators under local filtering operations. The tightness of the upper bounds have been investigated through detailed color noised quantum states. We have presented two classes of three-qubit states whose hidden genuine nonlocalities can be revealed by local filtering. Our results give rise to an operational method in investigating the genuine nonlocality for three-qubit mixed states.
Moreover, the method presented in this paper can also be used in optimizing the maximal quantum violations of other Bell-type inequalities for tripartite or multipartite quantum systems under local filtering.

\vspace{-1mm}  
\section*{Acknowledgment}
\vspace{-1mm}

\noindent{\bf Acknowledgments}\, \, This work is supported by NSFC (11775306, 11701568,
11675113), the Fundamental Research Funds for the Central
Universities (18CX02035A, 18CX02023A, 19CX02050A), Beijing Municipal Commission of Education under Grant No. KZ201810028042, Beijing Natural Science Foundation (Z190005), and Academy for Multidisciplinary Studies, Capital Normal University.

\end{document}